\documentclass[pdflatex,sn-mathphys-num]{sn-jnl}

\usepackage{graphicx}%
\usepackage{multirow}%
\usepackage{amsmath,amssymb,amsfonts}%
\usepackage{amsthm}%
\usepackage{mathrsfs}%
\usepackage[title]{appendix}%
\usepackage{xcolor}%
\usepackage{textcomp}%
\usepackage{manyfoot}%
\usepackage{booktabs}%
\usepackage{algorithm}%
\usepackage{algorithmicx}%
\usepackage{algpseudocode}%
\usepackage{listings}%
\usepackage{amsthm}

\usepackage{graphicx}
\usepackage{dcolumn}
\usepackage{bm}
\usepackage{enumitem}

\usepackage{booktabs}
\usepackage{setspace} 
\usepackage{algorithm}
\usepackage{algpseudocode}
\usepackage{xcolor}
\usepackage{tikz}
\usepackage{circuitikz}
\usetikzlibrary{decorations.pathreplacing}  
\usepackage{multirow,array}

\usepackage{adjustbox,tabularx}
\usepackage{subcaption}   
\usepackage{hyperref}
\hypersetup{
    colorlinks=true,
    linkcolor=blue,
    citecolor=blue,
    filecolor=blue,
    urlcolor=blue
}

\theoremstyle{thmstyleone}%
\newtheorem{theorem}{Theorem}
\newtheorem{proposition}[theorem]{Proposition}%

\theoremstyle{thmstyletwo}%
\newtheorem{remark}{Remark}%
\theoremstyle{thmstylethree}%
\newtheorem{definition}{Definition}%

\def\>{\ensuremath{\rangle}}
\def\<{\ensuremath{\langle}}

\usepackage{dsfont}

\begin{document}

\title[Article Title]{Validating a Koopman-Quantum Hybrid Paradigm for Diagnostic Denoising of Fusion Devices}


\author[1]{\fnm{Tie-Jun} \sur{Wang}}

\author[2]{\fnm{Run-Qing} \sur{Zhang}}

\author[2]{\fnm{Ling} \sur{Qian}}

\author[3]{\fnm{Yun-Tao} \sur{Song}} 

\author*[3]{\fnm{Ting} \sur{Lan}} \email{lanting@ipp.ac.cn}

\author*[3]{\fnm{Hai-Qing} \sur{Liu}} \email{hqliu@ipp.ac.cn}

\author*[4,5]{\fnm{Keren} \sur{Li}} \email{likr@szu.edu.cn}

\affil[1]{\orgdiv{School of Physical Science and Technology}, \orgname{Beijing University of Posts and Telecommunications}, \orgaddress{\city{Beijing}, \postcode{100876}, \country{China}}}
\affil[2]{\orgname{China Mobile (Suzhou) Software Technology Company Limited}, \city{Suzhou}, \postcode{215163}, \country{China}}
\affil[3]{\orgdiv{Institute of Plasma Physics}, \orgname{Chinese Academy of Sciences}, \orgaddress{\city{Hefei}, \postcode{230031}, \state{Anhui}, \country{China}}}
\affil[4]{\orgdiv{College of Physics and Optoelectronic Engineering}, \orgname{Shenzhen University}, \orgaddress{\city{Shenzhen}, \postcode{518060}, \state{Guangdong}, \country{China}}}
\affil[5]{\orgname{Quantum Science Center of Guangdong-Hong Kong-Macao Greater Bay Area (Guangdong) }, \orgaddress{\city{Shenzhen}, \postcode{518045}, \state{Guangdong}, \country{China}}}


\abstract{The potential of Quantum Machine Learning (QML) in data-intensive science is strictly bottlenecked the difficulty of interfacing high-dimensional, chaotic classical data into resource-limited, noisy quantum processors. To bridge this gap, we introduce a physics-informed Koopman-Quantum hybrid framework, theoretically grounded in a representation-level structural isomorphism we establish between the Koopman operator, which linearizes nonlinear dynamics, and quantum evolution. Based on this theoretical foundation, we design a realizable NISQ-friendly pipeline: the Koopman operator functions as a physics-aware "data distiller," compressing waveforms into compact, "quantum-ready" features, which are subsequently processed by a modular, parallel quantum neural network. We validated this framework on 4,763 labeled channel sequences from 433 discharges of the tokamak system. The results demonstrate that our model achieves 97.0\% accuracy in screening corrupted diagnostic data, matching the performance of state-of-the-art deep classical CNNs while using orders-of-magnitude fewer trainable parameters. This work establishes a practical, physics-grounded paradigm for leveraging quantum processing in constrained environments, offering a scalable path for quantum-enhanced edge computing.}

\maketitle

\section{Introduction}\label{sec1}

The advent of quantum machine learning (QML) has ignited transformative prospects for data-intensive scientific discovery, promising exponential accelerations in pattern recognition and complex correlation analysis~\cite{biamonte2017quantum, Schuld2015introduction, farhi2018classification, liu2018quantum}.  An $n$-qubit quantum neural network (QNN) operates in a $2^n$-dimensional Hilbert space, theoretically allowing it to capture complex correlations in a way that is fundamentally different from classical tensor operations~\cite{benedetti2019parameterized, huang2021power}. Recent studies suggest that QNNs can exhibit strong generalization capabilities even with limited data~\cite{wei2023quantum,ullah2024quantum}. 
However, the practical deployment of QML faces a foundational paradox: while quantum systems inherently operate in exponentially large Hilbert spaces, ideally suited for high-dimensional data processing, contemporary Noisy Intermediate-Scale Quantum (NISQ) hardware remains severely constrained in both qubit count and coherence time~\cite{li2024learning, mcclean2018barren, stilck2021limitations}. This creates an acute input bottleneck: how to encode high-dimensional, chaotic, and noise-contaminated classical data into the limited quantum register without information loss, while simultaneously ensuring trainability and interpretability~\cite{schuld2021effect, huang2021power, cerezo2022challenges}? 
This challenge is particularly critical for edge-computing scenarios where quantum resources are at a premium, yet real-time processing of massive data streams is imperative—such as in fusion plasma control~\cite{kates2019predicting}, autonomous systems, or environmental monitoring~\cite{karniadakis2021physicsinformed}.

Magnetic confinement fusion (MCF) diagnostics epitomize this challenge. Next-generation devices like ITER~\cite{creely2020overview} and CFETR~\cite{song2014concept} will generate large scale diagnostic streams per pulse; yet, a significant fraction of these signals is corrupted by electromagnetic interference, sensor anomalies, and mechanical vibrations~\cite{Li2018An, kube2022near}. Consequently, automated diagnostic screening, filtering out such stochastic corruptions, is a prerequisite for reliable physics analysis. However, traditional rule-based filtering fails to adapt to novel plasma regimes~\cite{Ferreira2018JETTomography}, while classical deep learning methods, despite their power, require massive parameterization ($\sim 10^4\text{--}10^5$ parameters) and suffer from data inefficiency and poor interpretability~\cite{Guo2021EASTLSTM,LAN2019159}. Given the promise of QML as a lightweight alternative, these limitations highlight a broader need: a computational framework capable of bridging the representational gap between complex classical data and resource-constrained quantum processors, while preserving the physics-driven interpretability essential for scientific validation~\cite{Huang2022Quantum}.

Here, we introduce a physics-informed classical-quantum hybrid interface that fundamentally rethinks how classical data can be prepared for quantum advantage in the NISQ era. Our approach is anchored in a deep mathematical insight: the Koopman operator~\cite{mezic2013analysis,brunton2021modern}, which linearizes nonlinear dynamical systems in an infinite-dimensional observable space, shares a structural isomorphism with the unitary evolution of quantum states in Hilbert space. 
By formalizing this isomorphism, we construct a principled two-stage pipeline: first, the Koopman operator acts as a physics-aware data distiller, compressing raw high-dimensional time-series into a low-dimensional subspace of interpretable features; second, these ``quantum-ready" features are processed by a modular, shallow variational quantum circuit designed specifically for NISQ constraints~\cite{Cerezo2021Variational,Bharti2022Noisy}.

Beyond its application to fusion diagnostics, this framework establishes a generalizable paradigm for quantum-enhanced edge computing~\cite{Furutanpey2023Architectural}. It demonstrates how domain-specific physical modeling (Koopman theory) can perform intelligent dimensionality reduction and feature extraction, effectively preparing classical big data for efficient quantum co-processing. This approach is applicable to any domain where data are high-dimensional, dynamically complex, and governed by partial physical laws—such as turbulence modeling in fluid dynamics, multi-scale climate simulations, or high-frequency financial time-series analysis~\cite{Lubasch2020Variational, Herman2023Quantum}. Our work thus addresses a critical question: how can we effectively ``feed" and leverage quantum co-processors with real-world classical data, especially in an era where quantum hardware resources will remain limited.

To validate the proposed framework, we implemented a NISQ-feasible hybrid architecture combining Koopman-based feature distillation with parallel quantum operations. Tested on 4,763 labeled channel sequences from 433 discharges of tokamak~\cite{Xie2025Neural}, our model achieves a state-of-the-art anomaly detection accuracy of $\sim 97.0\%$, matching deep classical convolutional neural networks (CNNs) while demonstrating an remarkable orders-of-magnitude reduction in parameter complexity. These findings provide compelling empirical evidence for the viability of quantum-enhanced computing in data-intensive scientific workflows.

This work thus transcends its immediate application to fusion diagnostics. It presents a validated pathway for overcoming the quantum input bottleneck, establishing a new paradigm where physical priors and quantum processing synergistically enable efficient, interpretable machine learning on classical big data, which is a crucial step toward realizing practical quantum advantage in the NISQ era and beyond.

\section{Result}
\label{sec:result}

\subsection{Problem Description and Structural Isomorphism}

\begin{figure}[!ht]
    \centering
    \includegraphics[width=1\linewidth]{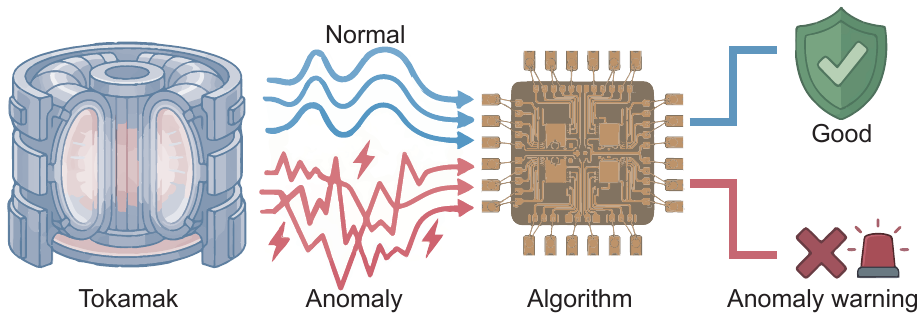}
    \caption{Overview of the diagnostic screening logic. 
     The tokamak environment generates signal streams where valid physical data (Normal, blue) is frequently intermixed with stochastic corruptions (Anomaly, red).
     The objective is to map these high-dimensional inputs to a binary decision, distinguishing scientifically usable data from corrupted ``Anomaly Warnings'' to protect downstream analysis.}    \label{fig:problem}
\end{figure}

\textbf{Diagnostic Challenges and Mathematical Formulation.} 
Next-generation MCF devices, such as ITER and CFETR, face formidable data challenges in real-time diagnostics. Taking POlarimeter-INTferometer system (POINT) system on the Experimental Advanced Superconducting Tokamak (EAST) as an example, its measured line-integrated electron density and Faraday rotation angle are critical parameters for plasma control and physics analysis. Each discharge produces multi-channel time-series data, denoted as \(\mathbf{x} = \{x_{t}\}_{t =0}^{T-1}\), where \(T\) is the number of time steps. In the harsh electromagnetic environment of a tokamak, the raw diagnostic signal stream is a mixture of valid physical signals and corruptions caused by fringe jumps, electromagnetic interference, and mechanical vibrations (Fig.~\ref{fig:problem}). These anomalies can render the signal scientifically unusable.

Thus, the objective of this study is to construct a binary classification function \(f: \mathbb{R}^T \to \{0,1\}\), acting as an automated gatekeeper that maps a standardized time-series segment $\mathbf{x}$ to a label $y$. Here, $y = 1$ denotes a physically consistent valid discharge (Normal), and $y = 0$ denotes a corrupted discharge (anomaly). Because plasma behavior evolves dynamically across different confinement modes, and the frequency content of legitimate physical events (e.g., edge localized modes) often overlaps with that of noise, simple threshold- or spectral-based filters are inadequate for this task.

\textbf{Theoretical Core: Koopman--Quantum Structural Isomorphism.}
To bridge the bottleneck between high-dimensional, chaotic classical data and quantum processors, we introduce a physics-informed classical--quantum hybrid interface.
The core insight is a representation-level structural isomorphism between Koopman evolution of observables and quantum evolution in Hilbert space, which is formalized in Methods (Sec.~\ref{method:Koopman-Quantum Isomorphism}) and sketched as follows.

Let $\{U_K^t\}$ be the Koopman operator family acting on observables, and let $\{\mathcal{U}^t\}$ denote a quantum evolution family on a Hilbert space.
Both frameworks realize the same time-translation algebra,
\begin{equation}
U_K^{t+s}=U_K^tU_K^s,
\qquad
\mathcal{U}^{t+s}=\mathcal{U}^t\mathcal{U}^s,
\label{eq:time-translation}
\end{equation}
which justifies embedding Koopman-linearized features into quantum Hilbert space for compatible processing (unitary for closed systems, and compatible via dilation for effective non-unitary dynamics).

In practice, we estimate a finite-dimensional Koopman generator $A$ such that $\dot{\mathbf{v}}(t)\approx A\mathbf{v}(t)$, yielding $\widehat{\mathbf{v}}(t)=e^{At}\mathbf{v}(0)$, which is a Koopman reconstruction of $\mathbf{x}$ [See Appendix \ref{app:koopman_math} for mathematical details]. An anomaly-sensitive residual vector is thus
\[
\mathbf{r}(t):=\dot{\mathbf{v}}(t)-A\mathbf{v}(t).
\]
Crucially, under a standard discretization, the residual energy admits a quadratic form
\begin{equation}
E(t):=\|\mathbf{r}(t)\|_2^2 \approx \tilde{\mathbf{v}}(t)^\dagger \mathcal{O}_{\mathrm{res}} \tilde{\mathbf{v}}(t),
\label{eq:residual-energy}
\end{equation}
where $\mathcal{O}_{\mathrm{res}}\succeq 0$ is determined by the difference scheme.
Eq.~\eqref{eq:residual-energy} is formally analogous to a quantum expectation value $\langle \psi | \hat{O} | \psi \rangle$: under a standard state encoding of $\tilde{\mathbf{v}}(t)$, the quadratic form can be represented as the expectation value of a corresponding observable acting on the encoded state.
This mapping is not merely heuristic; in Methods (Sec.~\ref{method:Koopman-Quantum Isomorphism}) we make the correspondence explicit in our restricted sense, including an isometric embedding between Koopman features and quantum states on the chosen finite feature subspace.
Consequently, Koopman reconstruction and residual statistics provide a physics-informed, ``quantum-ready'' interface suitable for processing by parameterized quantum circuits.

\subsection{Framework of the Classical-Quantum Hybrid Architecture}
To operationalize the structural isomorphism described above, we introduce the Koopman-PQNN architecture. As illustrated in Fig.~\ref{fig:pipeline}(a), this end-to-end pipeline integrates domain-specific physical modeling with quantum co-processing to bridge the gap between high-dimensional plasma data and NISQ constraints. The architecture comprises three functionally distinct yet mathematically coupled modules: a classical Koopman Embedding layer acting as a physics-informed filter, a Parallel Quantum Neural Network (PQNN) for non-linear feature processing in Hilbert space, and a measurement-based anomaly detection head. The specific implementation of each stage is detailed below (see Fig.~\ref{fig:pipeline}(b)-(c)).

\begin{figure}[!ht]
    \centering
    \includegraphics[width=0.95\linewidth]{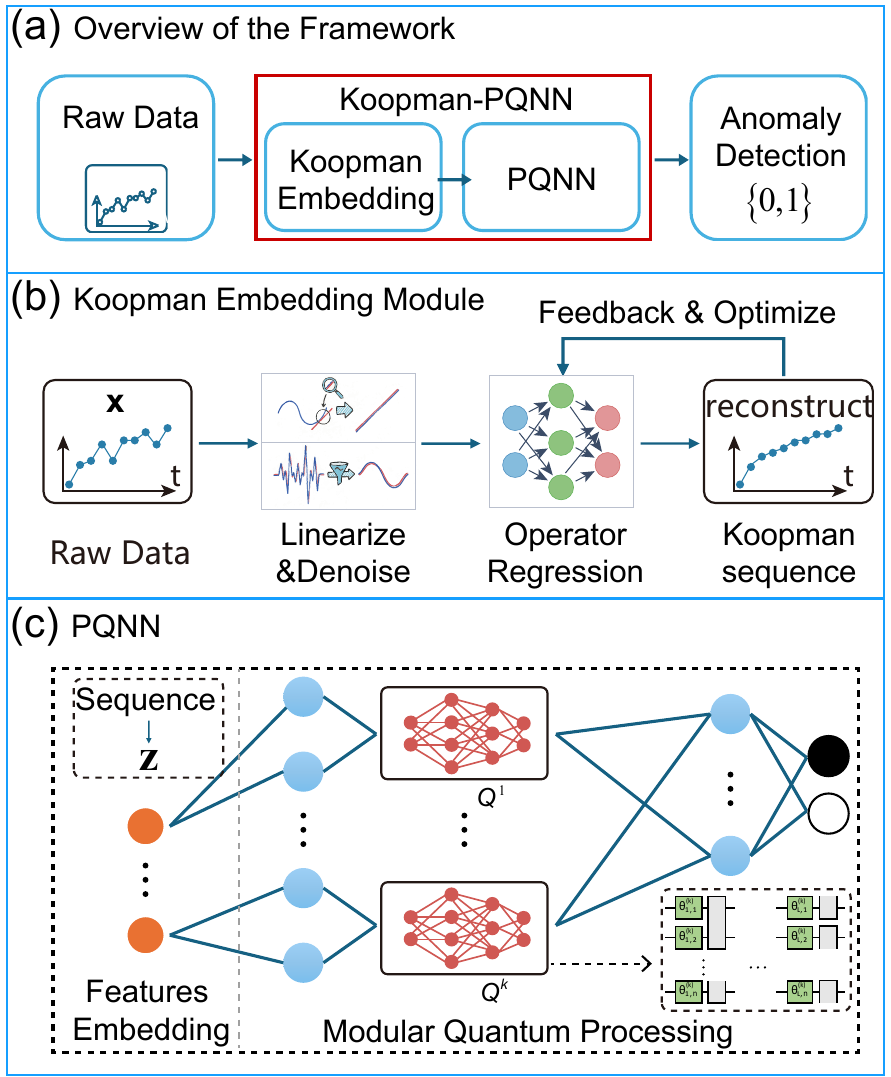}
    \caption{Architecture of the Koopman-PQNN Framework.  
(a) Overview: The system ingests raw multi-channel diagnostics, processes them through the proposed hybrid architecture, and outputs a binary validity score.
(b) Koopman Embedding Module: Raw chaotic signals are lifted into a linear space and reconstructed. 
(c) PQNN: To handle NISQ constraints, the feature vector is encoded and split across parallel Parameterized Quantum Circuits (PQCs), where each sub-circuit evolves features through variational layers.}
\label{fig:pipeline}
\end{figure}

\textbf{Koopman Embedding Module}: For a discrete-time system, the Koopman operator advances observables linearly. We estimate a finite-rank approximation of this operator to isolate dominant spectral modes. As shown in Fig.~\ref{fig:pipeline}(b), the feature extraction proceeds in three steps [See Appendix ~\ref{app:data} for details]:
(1) Lifting: Delay embedding via a Hankel matrix~\cite{mauroy2020koopman} projects the time series into a high-dimensional observable space.
(2) Spectral Decomposition: Singular Value Decomposition is performed on the Hankel matrix. The top $r$ singular vectors are retained to define a Koopman-invariant subspace, yielding the dominant dynamic modes $\mathbf{v}(t)$.
(3) Residual Characterization: Time derivatives $\dot{\mathbf{v}}(t)$ are estimated via finite differences, and a linear operator $A$ is regressed such that $A\mathbf{v}(t) \approx \dot{\mathbf{v}}(t)$. The reconstruction residuals—representing the non-linear deviations not captured by $A$—are then analyzed.
Finally, scalar statistics of these residuals (e.g., moments, skewness, kurtosis) are computed to form the compact feature vector $\mathbf{z} \in \mathbb{R}^m$. These Koopman-derived features preserve essential dynamics while reducing dimensionality by orders of magnitude~\cite{LAN2019159}, providing noise-robust inputs compatible with linear quantum operations~\cite{zhang2019dynamics}.

\textbf{Parallel Quantum Neural Network}: To handle both unitary and non-unitary cases while leveraging current NISQ devices, we adopt a modular design [See Fig.~\ref{fig:pipeline} (c) and Appendix~\ref{app:PQNN} for details]. The input features \(\mathbf{z}\in\mathbb{R}^{m}\) are mapped $ \mathbf{h} = W_{enc}\,\mathbf{z} + \mathbf{b}$, where \(m\) denotes the number of statistical features.
Here \(W_{enc}\in\mathbb{R}^{(nk)\times m}\) and \(\mathbf{b}\in\mathbb{R}^{nk}\) are randomly initialized and can be fixed during the entire training. \(k\) denotes the number of parallel sub-circuits, and \(n\) is the parameters per circuit. 
The projected vector $\mathbf{h}$ is then partitioned into $k$ sub-vectors $[\mathbf{h}^{(0)};\,\dots;\,\mathbf{h}^{(k-1)}]$. Each sub-vector $\mathbf{h}^{(j)} = [\phi_1^{(j)}, \dots, \phi_n^{(j)}]^T$ serves as the encoding angles for the $j$-th sub-circuit $Q^{(j)}$.
The quantum processing within each sub-circuit $Q^{(j)}$ consists of two stages:
(1) Angle Encoding: The classical information is embedded into the quantum state via single-qubit Pauli-$Y$ rotations:
$R_{y}(\phi^{(j)}_{i}) = \exp\bigl(-i\,\sigma_{y}\,\phi^{(j)}_{i}\bigr), i \in \{1,\dots,n\}$.
(2) Variational Evolution: $L$ layers of parameterized quantum gates are applied (Fig.~\ref{fig:pipeline}(c)), comprising fixed entangling operations and trainable local rotations.
This divide-and-conquer strategy not only captures non-trivial correlations via quantum entanglement but also significantly reduces the circuit depth, ensuring feasibility on near-term devices.

\textbf{Anomaly Detection}: To extract classical information, we perform Pauli-\(Z\) measurements on all qubits in each sub-circuit. This yields a local expectation vector \(\mathbf{e}^{(j)}\in \mathbb{R}^{n}\) for the \(j\)-th sub-circuit. 
These local vectors are concatenated into a global quantum feature vector, \(\bigl[\mathbf{e}^{(0)};\,\dots;\,\mathbf{e}^{(k-1)}\bigr]\).
Finally, a classical output layer with fixed weights aggregates these quantum features to estimate the anomaly probability and hence \(\hat{y} \in [0,1]\).
In this work, only the parameters within the variational quantum layers are optimized, ensuring the quantum Hilbert space is effectively explored to resolve non-linear decision boundaries.

\subsection{Experimental Setup and Performance Comparison}

\begin{figure}[!ht]
    \centering
    \includegraphics[width=0.6\linewidth]{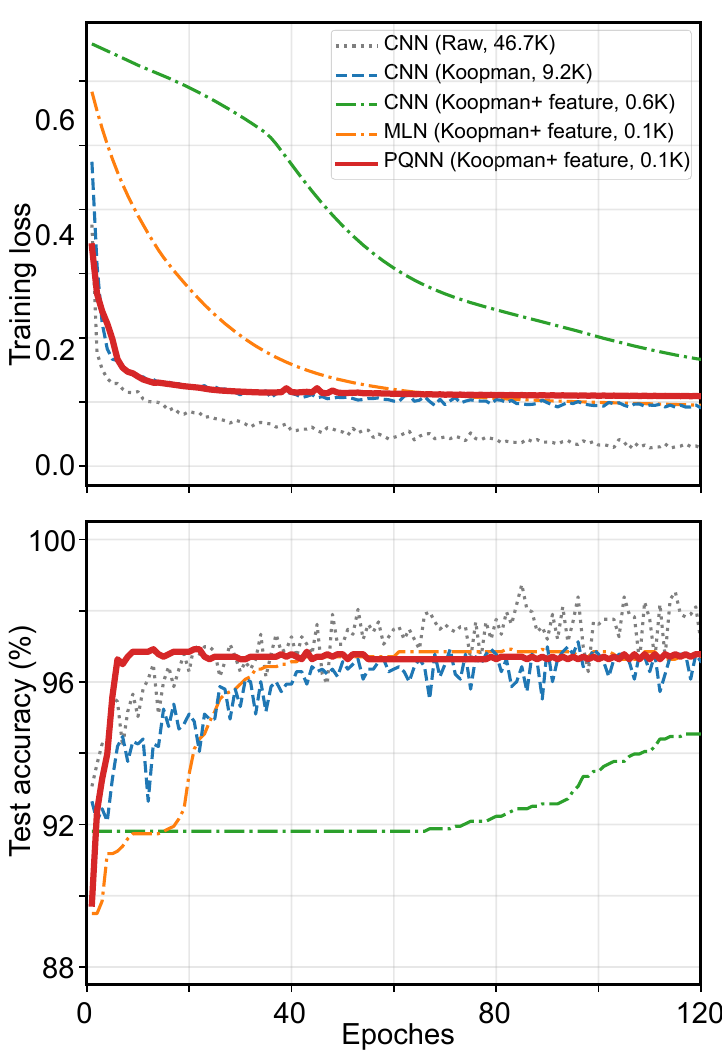}
    \caption{Comparison of training loss and test accuracy. The plots benchmark the proposed Koopman-PQNN (Red solid line) against four baselines: Raw-Data CNN (Grey dotted line), CNN on Koopman sequences (Blue dashed line), CNN on Koopman features (Green dash-dot line), and Classical  multi-layer network (MLN) on Koopman features (Orange dash-dot line).
    (a) Training Loss: The curves demonstrate that all models successfully minimize the objective function and achieve stable convergence during the training process. (b) Test Accuracy: The evolution of testing accuracy indicates that while the computationally heavy Raw-Data CNN achieves a slightly superior performance ceiling ($\sim 98.0\%$), the proposed Koopman-PQNN and other physics-informed baselines all consistently converge to a comparable and robust accuracy level of $\sim 97.0\%$.}
    \label{fig:result}
\end{figure}

We validated the proposed framework on the EAST POINT diagnostic dataset, which comprises 11 channels of data from 433 discharges, yielding \(N = 4,763\) valid pulse sequences after filtering. The data were randomly split into a training set (70\%, \(N_{\mathrm{train}} = 3,334\)) and a test set (30\%, \(N_{\mathrm{test}} = 1,429\)). We constructed a three-stage processing pipeline: starting from high-frequency raw waveforms, projecting onto the invariant subspace to obtain Koopman sequences of residuals, and finally extracting statistical descriptors from residuals, hereafter referred to as Koopman features. Our Koopman-PQNN specifically operates on a compact 6-dimensional feature vector \(\mathbf{z}\) derived from the residuals [See Appendix~\ref{app:data} for details]. Prior to quantum encoding, these vectors were min-max scaled to the \([0, \pi]\) interval to align with the periodic sensitivity of quantum rotation gates.

To rigorously position our method within the landscape of fusion diagnostics, we benchmarked Koopman-PQNN against four representative paradigms. Comparative results are summarized in Fig.~\ref{fig:result}:

\textbf{Accuracy Ceiling: Deep CNN on Raw Data}. As a computational performance upper bound, we trained a deep 1D CNN~\cite{Li2022Convolutional} directly on the standardized raw time-series. This brute-force approach, leveraging deep layers to automatically learn filters from full-resolution sequences, achieved the highest test accuracy of 98.0\%. However, this marginal 1\% gain over other methods comes at a prohibitive cost: the model requires approximately 46.7k trainable parameters and processes high-dimensional inputs, creating a computational bottleneck unsuitable for real-time edge deployment. Notably, when constraining the Raw-Data CNN to a smaller parameter budget (\(\sim\)9.2k), its accuracy dropped significantly to 94.0\%, highlighting heavy parameterization required to extract features from raw data.

\textbf{Koopman-PQNN}. In contrast, our proposed Koopman-PQNN achieved a test accuracy of \(\sim\)97.0\%. While slightly below the deep learning ceiling (1\% gap), this was accomplished using only $0.1$k parameters—a remarkable reduction of two orders of magnitude compared to the Raw-Data CNN. This result demonstrates that the Koopman operator acts as a highly efficient physics filter, capturing the vast majority of relevant dynamics. The PQNN effectively matches the performance of other feature-based methods but with superior parameter efficiency, offering the optimal balance for resource-constrained diagnostic systems.

\textbf{Intermediate Baselines}. We also evaluated intermediate approaches to isolate the contributions of different modules. \textbf{Applying a CNN to the Koopman sequences of residuals} (9.2k parameters) yielded \(\sim\)97.0\% accuracy, confirming that Koopman linearization simplifies the learning task, allowing a smaller CNN to perform well, albeit still requiring a significant number of parameters. \textbf{Training a CNN model on the 6D Koopman features} (0.6k parameters) also converged to \(\sim\)97.0\% test accuracy. Crucially, we employed a classical multi-layer neural network (MLN) as a baseline for the quantum architecture. Mathematically, the unitary transformations inherent to quantum circuits constitute a strict subset of the general linear transformations accessible to classical layers. Consequently, the classical MLN serves as an unconstrained upper bound for expressivity. The fact that our unitarity-constrained PQNN matches the performance of the general MLN validates that the proposed quantum ansatz possesses comparable expressivity to capture the decision boundary without requiring the full parameter space of classical networks.

The training loss and test accuracy curves in Fig.~\ref{fig:result} show that all models successfully minimized the objective function and achieved stable convergence during training. While the computationally intensive Raw-Data CNN reached a slightly superior performance ceiling (\(\sim\)98.0\%), the proposed Koopman-PQNN and other physics-informed baselines all consistently converged to a comparable and robust accuracy level of \(\sim\)97.0\%.

\subsection{Ablation Study and Mechanism Analysis}
To deconstruct the performance gains and understand the underlying mechanisms of our hybrid architecture, we conducted a multi-stage ablation study, focusing on two critical phases: the impact of residuals via Koopman linearization, and the expressivity of the quantum ansatz compared to classical function approximators in the compressed Koopman feature space.

\subsubsection{Impact of Residuals via Koopman Linearization}
We first analyze the transition from raw diagnostic data to the residual statistics derived from Koopman analysis. The comparison between the Raw-Data CNN (98.0\%, 46.7k params) and the Koopman-based CNN (97.0\%) reveals a crucial mechanism introduced by the hybrid embedding.
While relying on residuals incurs a marginal accuracy loss (1\%) due to the information compression, the Koopman operator effectively acts as a physics-informed background subtractor. By explicitly modeling and removing the dominant, predictable linear dynamics, the resulting residuals isolate the pure non-linear deviations and stochastic bursts characteristic of plasma disruptions. Unlike raw data, where the model must learn both the global pulse evolution and the local anomalies, the residual space allows the classifier to focus solely on the discriminative innovation signal. This significantly reduces the representational complexity, enabling a lightweight network to capture the decision boundary with high efficiency, verifying that the residuals successfully distill the essential entropic information required for screening.

\subsubsection{Quantum Expressivity and Classical Equivalence}
This section investigates the mapping from statistical features to classification labels. Driven by the qubit count limitations of current hardware, we compressed the sequences into 6-dimensional statistical descriptors. Here, we address a fundamental theoretical question: Is the quantum circuit strictly necessary for this low-dimensional manifold?

\textbf{Parameter Efficiency vs. CNNs}. We isolate the efficiency of the classifier operating on the Koopman feature space. While a CNN trained on these features can achieve high accuracy, it typically requires 0.6k parameters to construct the decision hyperplanes. In contrast, our PQNN achieves comparable state-of-the-art accuracy with only $\sim$0.1k trainable parameters. This demonstrates that the quantum ansatz provides a significantly more compact and parameter-efficient representation for physics-informed features.

\textbf{Mathematical Subset and Convergence Dynamics}. We compare the PQNN against a classical MLN to probe the theoretical bounds of our approach. Mathematically, the unitary transformations executed by a quantum circuit constitute a strict subset of the general linear transformations allowed by the weight matrices of a classical MLN. Consequently, a classical MLN is theoretically capable of approximating any function the quantum circuit can produce. Our experimental results confirm this inclusion relation: an MLN trained on the same 6D Koopman features eventually achieved a comparable test accuracy of $\sim$97.0\%.

However, this mathematical containment overlooks crucial practical distinctions. First, in training dynamics, the PQNN converged to its high-performance solution significantly faster and with greater stability than the MLN (Fig.~\ref{fig:result}a). This suggests that the unitary inductive bias—the constraint that the circuit must perform a norm-preserving evolution—sculpts a more structured optimization landscape. The quantum model acts as a ``fast learner,'' efficiently capturing essential feature correlations with far fewer parameter updates.

\textbf{The Scalability Argument}. While classical MLNs suffice for this specific, short-pulse validation case, the core argument for our hybrid paradigm hinges on scalability to high-complexity regimes. Future fusion diagnostics (e.g., ITER steady-state operations) will involve longer pulse durations and multi-modal data fusion, necessitating high-dimensional feature spaces. In such regimes, the expressive power of a quantum circuit, driven by entanglement and superposition, scales exponentially with the number of qubits. This offers a pathway to model complexities that would require prohibitively large classical networks. Therefore, our present work serves as a critical proof-of-principle: it demonstrates that the Koopman-quantum pipeline is not only performant but is architected for scalability, establishing a foundation for future applications where classical methods may hit fundamental efficiency barriers.

\begin{figure}[!ht]
    \centering
    \includegraphics[width=1\linewidth]{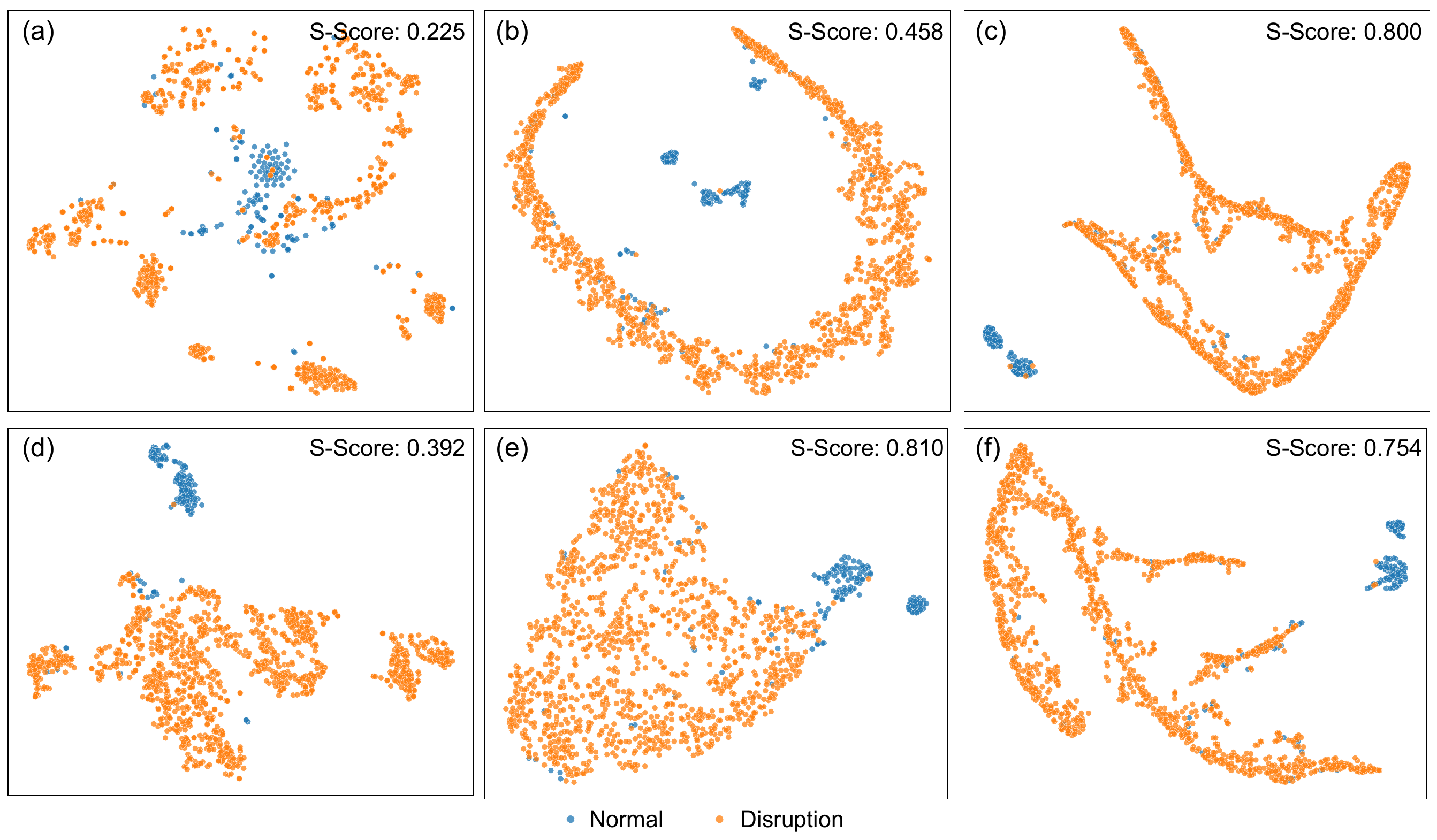}
    \caption{Evolution of feature separability across different paradigms.
    The t-SNE projections visualize the test set samples (Orange: Stable, Blue: Anomaly) with their corresponding S-Score.
    (a) Raw Data (S=0.225) and (b) Koopman Sequences (S=0.458) illustrate the intrinsic difficulty of the task with heavy class overlap.
    (d) Raw-Data CNN (S=0.392) successfully separates classes but lacks cluster cohesion.
    (e) Koopman Seq CNN establishes the topological ceiling with the highest separation (S=0.810).
    (c) Our Method (Koopman-PQNN) achieves a remarkable S-Score of 0.800, significantly outperforming the classical (f) Feature CNN (S=0.754).
    This demonstrates that the lightweight quantum circuit matches the high-quality clustering of deep Koopman baselines while drastically outperforming the loose feature representation of the Raw-Data CNN.}
    \label{fig:tsne}
\end{figure}

\subsubsection{Visual Interpretability and Topological Distinctness}
To corroborate the classification metrics and intuitively understand how the proposed architecture transforms the data manifold, we employed t-Distributed Stochastic Neighbor Embedding (t-SNE) to project high-dimensional feature representations onto a 2D plane (Fig.~\ref{fig:tsne}, Appendix \ref{app:visualization})~\cite{Laurens2008Visualizing}. Furthermore, to move beyond qualitative assessment, we calculated the Silhouette Score (S-Score) for each projection. The S-Score ranges from -1 to 1, where a value closer to 1 indicates perfectly separated and cohesive clusters, providing a rigorous quantitative metric for decision boundary quality~\cite{Rousseeuw1987Silhouettes}.

\textbf{Intrinsic Data Structure}. The projection of standardized Raw Data [Fig.~\ref{fig:tsne}(a)] reveals a highly entangled topology with no clear separation margin between the Stable (orange) and Anomaly (blue) classes. This visual ambiguity is quantitatively confirmed by a low S-Score of 0.225, explaining why traditional threshold-based methods fail on raw diagnostics. The residuals of Koopman sequences [Fig.~\ref{fig:tsne}(b)] exhibit a slightly improved local structure (S-Score: 0.458) due to the filtering of high-frequency stochastic noise. However, the classes remain linearly inseparable, indicating that while Koopman linearization simplifies the dynamics, it does not by itself solve the classification problem without a non-linear classifier.

\textbf{Classical Benchmarks}. Comparing classical deep learning baselines reveals a fundamental distinction in feature learning. The Raw-Data CNN [Fig.~\ref{fig:tsne}(d)] successfully separates the classes (as evidenced by its 98\% accuracy), but the clusters remain loose and dispersed, resulting in a relatively low S-Score of 0.392. This suggests that while the brute-force CNN finds a valid decision boundary, it does not necessarily map the data into a compact, structurally cohesive manifold. In contrast, the Koopman sequence CNN [Fig.~\ref{fig:tsne}(e)] achieves the highest clustering quality with an S-Score of 0.810. By operating on the physically linearized Koopman sequences, this model forces the classes into tight, distinct islands. Similarly, the CNN on Koopman features [Fig.~\ref{fig:tsne}(f)] achieves a high S-Score of 0.754, confirming that the statistical moments extracted from the Koopman space offer a naturally cohesive representation.

\textbf{Quantum Topological Distinctness}. Crucially, our proposed Koopman-PQNN [Fig.~\ref{fig:tsne}(c)] generates a highly distinct decision landscape with an S-Score of 0.800. This result is significant for two reasons: First, purely in terms of separation quality, the PQNN outperforms the classical feature CNN (0.800 vs. 0.754) and effectively matches the ceiling set by the Koopman sequence CNN (0.810). Second, it highlights a unique quantum advantage in representation learning: while the Raw-Data CNN requires 46k parameters to achieve separation, the PQNN achieves both separation and high cohesion using only $\sim$0.1k parameters.

\section{Discussion}
\label{sec:discussion_conclusion}
In this work, we have introduced and validated a novel, operator-theoretic quantum machine learning pipeline for the validation of MCF diagnostics. By bridging Koopman operator theory with quantum operations, we present the first application of a physics-informed quantum protocol to the cleaning of fusion experimental data.

\textbf{Performance and Efficiency Trade-off.}
Validated on 4,763 labeled channel sequences from 433 discharges of the EAST tokamak, the proposed Koopman-PQNN demonstrates a distinct advantage over existing paradigms. 
While achieving accuracy parity with state-of-the-art deep learning benchmarks (e.g., Raw-Data CNNs at $\sim 98.0\%$), our hybrid model distinguishes itself by superior efficiency and interpretability. 
Specifically, while the Raw-Data CNN suffers from loose feature clustering (S-Score $\approx 0.39$) despite its heavy parameterization ($\sim 46.7$k), the Koopman-PQNN achieves highly cohesive clusters (S-Score $\approx 0.80$) with two orders of magnitude fewer parameters ($\sim 0.1$k). 
This confirms that embedding physical laws into the learning process allows lightweight quantum circuits to serve as robust, efficient gatekeepers for scientific data, matching the performance of deep networks without their computational overhead.

\textbf{Mechanism.}
Our experimental results demonstrate that the Koopman-quantum hybrid architecture achieves an efficient functional equivalence. The Koopman operator acts as a physics-informed spectral sieve: rather than merely compressing data, it models and subtracts the dominant linear dynamics, thereby isolating the non-linear residuals that carry the critical signatures of plasma anomalies.  Subsequently, leveraging the structural correspondence established in our method, the variational quantum circuit processes these ``observable-like'' residual features. The entanglement layers effectively reconstruct the complex correlations within this purified feature space, resolving the decision manifold within the high-dimensional Hilbert space using minimal parameters.

The superiority of this division of labor is reflected in two key aspects: First, in terms of performance, our model matches the accuracy of deep convolutional networks with a 100-fold improvement in parameter efficiency. Second, in terms of representation quality, our model generates feature clusters with significantly higher topological compactness (S-Score 0.80 vs. 0.39). These findings confirm that we are not merely concatenating two modules, but rather realizing a synergistic enhancement: the classical component uses physical priors to filter out the trivial linear evolution, while the quantum component dedicates its scarce resources to analyzing the intractable non-linear deviations. This paradigm provides a clear pathway toward realizing quantum advantage on NISQ devices: identifying rare events by focusing quantum power solely on the physics that classical linear models fail to capture.

\textbf{Scalability and Future Outlook.}
Tailored to the constraints of the NISQ era, our architecture adopts a parallel modular design. This approach decouples the processing of multi-channel inputs, thereby avoiding the prohibitive coherence costs associated with large, fully entangled quantum registers. Consequently, the framework ensures linear scalability with constant circuit depth, offering a viable path to process the petabyte-scale diagnostic streams anticipated in future devices like ITER.

To our knowledge, this represents the first demonstration of a physics-aligned quantum pipeline processing real-world fusion experimental data. Looking ahead, future work will focus on deploying this architecture on cloud-access quantum hardware to validate its resilience against real device noise, and extending the Koopman-Quantum paradigm to multi-modal fusion diagnostics to support integrated, real-time plasma control.

\section{Methods $:$ Koopman-Quantum Isomorphism}  
\label{method:Koopman-Quantum Isomorphism}

This section formalizes the structural correspondence between Koopman operator evolution of classical observables and quantum dynamical evolution, which underpins our hybrid Koopman--quantum pipeline.

\begin{remark}
Throughout, ``Koopman--Quantum Isomorphism'' refers to a representation/structure-level correspondence: both theories provide linear operator representations of time translation, and finite-dimensional Koopman features admit an isometric embedding into a quantum Hilbert space for compatible processing.
This notion does not require a global conjugacy between the full Koopman operator and a specific physical Hamiltonian.
\end{remark}

\subsection{Structural Correspondence of One-Parameter Linear Operator Families}
\label{subsec:KQI-step1}

Let \(M\) be a smooth manifold and consider an autonomous dynamical system
\[
\dot{\mathbf{x}}=\mathbf{F}(\mathbf{x}),\qquad \mathbf{x}\in M,
\]
with flow \(\mathbf{S}^t:M\to M\) satisfying \(\mathbf{S}^{t+s}=\mathbf{S}^t\circ \mathbf{S}^s \) and \(\mathbf{S}^0=I\).

\begin{definition}[Koopman operator family]
Let \(\mathcal{F}\) be a suitable linear space of observables \(g:M\to\mathbb{C}\).
The Koopman operator family \(\{U_K^t\}_{t\in\mathbb{R}}\) is defined by
\[
(U_K^t g)(\mathbf{x}) := g(\mathbf{S}^t(\mathbf{x})),\qquad \forall g\in\mathcal{F}.
\]
\end{definition}

\begin{proposition}[Basic properties of \(U_K^t\)]
The Koopman operators are linear and satisfy the time-translation composition law \cite{mauroy2020koopman}:
\[
U_K^{t+s}=U_K^tU_K^s,\qquad U_K^0=I.
\]
If the flow is invertible for all \(t\), then \(U_K^{-t}=(U_K^t)^{-1}\) and \(\{U_K^t\}\) forms a one-parameter group.
\end{proposition}

In quantum mechanics, a closed system evolves according to
\[
i\hbar\,\partial_t|\psi(t)\rangle = H|\psi(t)\rangle,
\]
with unitary propagator \(\mathcal{U}^t=e^{-iHt/\hbar}\).

\begin{definition}[Quantum evolution operator family]
The quantum evolution operator family \(\{\mathcal{U}^t\}_{t\in\mathbb{R}}\) is defined as
\[
\mathcal{U}^t := e^{-iHt/\hbar}.
\]
\end{definition}

\begin{proposition}[Basic properties of \(\mathcal{U}^t\)]
The operators \(\mathcal{U}^t\) are linear and form a one-parameter unitary group:
\[
\mathcal{U}^{t+s}=\mathcal{U}^t\mathcal{U}^s,\qquad \mathcal{U}^0=I,\qquad (\mathcal{U}^t)^{-1}=\mathcal{U}^{-t}.
\]
\end{proposition}

\begin{theorem}[Common time-translation structure]
\label{thm:KQI-struct}
Both \(t\mapsto U_K^t\) and \(t\mapsto \mathcal{U}^t\) are homomorphisms from \((\mathbb{R},+)\) into the composition semigroup/group of linear operators on \(\mathcal{F}\) and \(\mathcal{H}\), respectively:
\[
U_K^{t+s}=U_K^tU_K^s,\qquad \mathcal{U}^{t+s}=\mathcal{U}^t\mathcal{U}^s.
\]
Hence, Koopman evolution of observables and quantum evolution of states share the same algebraic time-translation skeleton at the operator-family level.
\end{theorem}

\subsection{Spectral Correspondence of Exponential Dynamical Modes}
\label{subsec:KQI-step2}

\begin{definition}[Koopman eigenfunction]
A nonzero \(\phi_\lambda\in\mathcal{F}\) is a Koopman eigenfunction with eigenvalue \(\lambda\in\mathbb{C}\) if
\[
U_K^t\phi_\lambda = e^{\lambda t}\phi_\lambda,\qquad \forall t\in\mathbb{R}.
\]
\end{definition}

\begin{proposition}[Multiplicative closure]
If \(\phi_{\lambda_1}\) and \(\phi_{\lambda_2}\) are Koopman eigenfunctions, then their pointwise product \(\phi_{\lambda_1}\phi_{\lambda_2}\) is also a Koopman eigenfunction with eigenvalue \(\lambda_1+\lambda_2\).
\end{proposition}

\begin{definition}[Hamiltonian eigenstate]
An eigenstate \(|\psi_n\rangle\in\mathcal{H}\) of \(H\) satisfies
\[
H|\psi_n\rangle = E_n|\psi_n\rangle,\qquad E_n\in\mathbb{R},
\]
and evolves as \(\mathcal{U}^t|\psi_n\rangle = e^{-iE_n t/\hbar}|\psi_n\rangle\).
\end{definition}

\begin{remark}[Formal spectral identification]
Both Koopman and quantum evolutions diagonalize into exponential modes:
\(e^{\lambda t}\) for Koopman eigenfunctions and \(e^{-iEt/\hbar}\) for Hamiltonian eigenstates.
A convenient \emph{formal} identification is
\[
\lambda \;\leftrightarrow\; -\frac{i}{\hbar}E,
\]
which matches the exponential time dependence of individual modes.
\end{remark}

\begin{remark}[Complex Koopman spectra and non-unitary quantum realizations]
Koopman spectra are generally complex, and $\Re(\lambda)\neq 0$ corresponds to growth/decay in $e^{\lambda t}$.
This does not obstruct the present correspondence: such non-unitary evolution can be interpreted as
(i) an effective non-Hermitian generator $e^{-iH_{\mathrm{eff}}t}$ with $H_{\mathrm{eff}}\neq H_{\mathrm{eff}}^\dagger$, or
(ii) an open-system quantum channel whose eigenmodes exhibit exponential decay.

Importantly, non-unitary dynamics can still be realized using unitary quantum circuits through standard enlargements of the Hilbert space.
In particular, our recent experimental work~\cite{li2025realization} demonstrates that non-Hermitian and imaginary-time evolutions can be implemented on
today's hardware by decomposing the target non-unitary propagator via a truncated Cauchy-integral representation into a collection of Hermitian/unitary subtasks, which are executed in parallel and classically aggregated.
Therefore, the mapping between Koopman exponential modes and quantum(-compatible) dynamical generators should be understood as a mode-level correspondence, not restricted to strictly unitary closed-system dynamics.
\end{remark}

\subsection{Representational Mapping from Koopman Features to Quantum States}
\label{subsec:KQI-step3}

The operator-family correspondence motivates a feature-to-state representation: Koopman-derived features (ideal eigenfunctions or data-driven surrogates) are embedded into an \(n\)-qubit Hilbert space for quantum processing.

Let \( \mathcal{G} \subset \mathcal{F} \) be a finite-dimensional invariant subspace spanned by Koopman eigenfunctions \( \{\phi_{\lambda_i}\}_{i=1}^d \). Endow \( \mathcal{G} \) with the inner product
\[
\langle \phi_i, \phi_j \rangle_\mathcal{G} = \int_M \phi_i^*(x) \phi_j(x) \, d\mu(x),
\]
where \(\mu\) is an invariant measure on \(M\). Let \(\mathcal{H}\) be an \(n\)-qubit Hilbert space with \(\dim \mathcal{H} = 2^n \geq d\), equipped with the standard inner product \(\langle \cdot | \cdot \rangle\).

\begin{theorem}[Isometric embedding and mode-preserving realization]
\label{thm:isometric-embedding}
There exists a linear map \( \mathcal{R}: \mathcal{G} \to \mathcal{H} \) satisfying:
\begin{enumerate}
    \item \textbf{Isometry:} \(\langle \phi_i, \phi_j \rangle_\mathcal{G} = \langle \mathcal{R}(\phi_i) | \mathcal{R}(\phi_j) \rangle\) for all \(\phi_i, \phi_j \in \mathcal{G}\).
    \item \textbf{Mode preservation (effective generator):}
    for each Koopman eigenfunction \(U_K^t \phi_{\lambda_i} = e^{\lambda_i t}\phi_{\lambda_i}\), there exists a (possibly non-Hermitian) operator \(H_{\mathrm{eff}}\) acting on \(\mathrm{span}\{\mathcal{R}(\phi_{\lambda_i})\}_{i=1}^d\) such that
    \[
    e^{-iH_{\mathrm{eff}}t/\hbar}\,\mathcal{R}(\phi_{\lambda_i}) = e^{\lambda_i t}\,\mathcal{R}(\phi_{\lambda_i}),\qquad \forall t.
    \]
\end{enumerate}
\end{theorem}

\begin{proof}
Apply Gram--Schmidt orthogonalization to \(\{\phi_{\lambda_i}\}\) to obtain an orthonormal basis \(\{\hat{\phi}_i\}_{i=1}^d\) of \(\mathcal{G}\).
Define \(\mathcal{R}(\hat{\phi}_i)=|i\rangle\), where \(\{|i\rangle\}_{i=1}^d\) are orthonormal computational basis states in \(\mathcal{H}\), and extend \(\mathcal{R}\) linearly; the isometry property follows immediately.

For mode preservation, define \(H_{\mathrm{eff}}\) on \(\mathrm{span}\{|i\rangle\}_{i=1}^d\) by
\[
H_{\mathrm{eff}}|i\rangle := i\hbar\,\lambda_i\,|i\rangle,
\]
which may be non-Hermitian when \(\lambda_i\in\mathbb{C}\).
Then \(e^{-iH_{\mathrm{eff}}t/\hbar}|i\rangle = e^{\lambda_i t}|i\rangle\), implying the stated relation for \(\mathcal{R}(\phi_{\lambda_i})\).
\end{proof}

\begin{remark}[Compositionality]
The multiplicative closure of Koopman eigenfunctions provides a natural compositional rule on the classical side.
When compositional feature constructions are needed in quantum space, one may employ tensor-product embeddings by allocating additional registers.
We use this as motivation rather than a required axiom of the finite-dimensional embedding.
\end{remark}

This representation mapping provides a principled, structure-preserving embedding of classical dynamical modes into quantum Hilbert space, enabling the quantum circuit to process Koopman-linearized features while preserving their geometric structure.

\subsection{Bridging Theory and Implementation via Residual Observables and Statistical Descriptors}
\label{subsec:KQI-step4}

In our case, we construct a residual-based observable that is structurally aligned with the Koopman--quantum correspondence.

\paragraph{Finite-dimensional Koopman approximation.}
Given a time series \(\mathbf{x}\), we construct reduced Koopman coordinates \(\mathbf{v}(t)\in\mathbb{R}^m\) via delay embedding and SVD/DMD,
and fit a linear model
\[
\dot{\mathbf{v}}(t)\approx A\,\mathbf{v}(t).
\]
Define the residual vector and its magnitude
\begin{equation}
\label{eq:KQI-residual}
\mathbf{r}(t):=\dot{\mathbf{v}}(t)-A\mathbf{v}(t),\qquad R(t):=\|\mathbf{r}(t)\|_2.
\end{equation}
The residual \(R(t)\) quantifies local deviation from Koopman-linear predictability and is sensitive to anomalies/regime changes.

\paragraph{Residual as a linear (and quadratic) observation.}
In discrete time, \(\dot{\mathbf{v}}(t)\) is realized via the second-order central difference scheme
\[
\dot{\mathbf{v}}(t)\approx \frac{\mathbf{v}(t+\Delta t)-\mathbf{v}(t-\Delta t)}{2\Delta t}.
\]
Introduce the augmented state
\[
\tilde{\mathbf{v}}(t):=
\begin{bmatrix}
\mathbf{v}(t+\Delta t)\\ \mathbf{v}(t)\\ \mathbf{v}(t-\Delta t)
\end{bmatrix}\in\mathbb{R}^{3m},
\qquad
B:=
\begin{bmatrix}
\frac{1}{2\Delta t}I_m & -A & -\frac{1}{2\Delta t}I_m
\end{bmatrix}.
\]
Then
\begin{equation}
\label{eq:KQI-linear-r}
\mathbf{r}(t)\approx B\,\tilde{\mathbf{v}}(t),
\end{equation}
and the residual energy admits a quadratic form
\begin{equation}
\label{eq:KQI-energy}
E(t):=\|\mathbf{r}(t)\|_2^2 \approx \tilde{\mathbf{v}}(t)^\dagger (B^\dagger B)\tilde{\mathbf{v}}(t),
\end{equation}
i.e., \(E(t)\) is a positive semidefinite ``physical observation'' extracted from Koopman-consistent linear evolution under appropriate normalization and encoding.

\paragraph{Statistical descriptors as weighted sums of observations.}
Given a windowed sequence \(\{R(t)\}_{t=1}^T\) (or \(\{E(t)\}_{t=1}^T\)), we form a statistical descriptor vector \(\mathbf{z}\in\mathbb{R}^d\),
whose components can be written as empirical aggregations
\begin{equation}
\label{eq:KQI-stats}
z_k = \frac{1}{T}\sum_{t=1}^T w_k\!\big(R(t)\big)
\quad (\text{or } z_k=\frac{1}{T}\sum_{t=1}^T \tilde w_k\!\big(E(t)\big)),
\qquad k=1,\dots,d,
\end{equation}
for suitable weights \(w_k\).
Hence, \(\mathbf{z}\) is a summation of weighted physical observations derived from the Koopman residual.

\subsection{Summary of Proof of Koopman--Quantum Isomorphism}
\label{subsec:KQI-summary}

Therefore, we establish a principled Koopman--quantum correspondence at the level required by our hybrid pipeline by showing that:
(i) Koopman and quantum evolutions share the same time-translation operator-family structure;
(ii) their exponential mode decompositions admit a formal spectral correspondence, extendable to complex spectra via non-unitary quantum realizations;
(iii) finite-dimensional Koopman features admit an isometric embedding into a quantum Hilbert space with a mode-preserving (possibly non-Hermitian) effective generator; and
(iv) a residual-based observable provides a concrete bridge from Koopman approximation to implementable ``physical observations'' and their statistical descriptors.
\section{Data availability}
All data supporting the results of this study are available in the manuscript or the supplementary information. Additional data are available from the corresponding author upon reasonable request.


\section{Acknowledgements}
This work is supported in part by the National Natural Science Foundation of China under Grant No. 62471058,
and in part by the Fund of State Key Laboratory of Information Photonics and Optical Communications, Beijing University of Posts and Telecommunications, China (Grant No. IPOC2022ZT10). We thank the staff members at EAST in Hefei (https://cstr.cn/31130.02.EAST), for providing technical support and assistance in data collection and analysis.

\section{Author contributions}

T.W. and K.L. conceived the project and designed the numerical experiments. T.L., S.Y. and H.L. contributed the EAST experimental data and participated in the scenario analysis and physical interpretation. L.Q. and R.Z. provided useful insights for the
characterization and understanding of the Koopman-quantum hybrid paradigm. T.W. and K.L. wrote the paper, while all authors discussed the results and commented on the manuscript.

\section{Competing interests}
The authors declare no competing interests.

\section{Additional information}
Supplementary information is provided in Appendix A-C. 
Correspondence and requests for materials should be addressed to Keren Li.

\bibliography{main_qnn} 

@article{creely2020overview,
  title={Overview of the SPARC tokamak},
  author={Creely, AJ and Greenwald, Martin J and Ballinger, Sean B and Brunner, D and Canik, J and Doody, Jeffrey and F{\"u}l{\"o}p, T and Garnier, DT and Granetz, R and Gray, TK and others},
  journal={Journal of Plasma Physics},
  volume={86},
  number={5},
  pages={865860502},
  year={2020},
  doi={10.1017/S0022377820001257},
  publisher={Cambridge University Press}
}

@article{song2014concept,
  title={Concept design of CFETR tokamak machine},
  author={Song, Yun Tao and Wu, Song Tao and Li, Jian Gang and Wan, Bao Nian and Wan, Yuan Xi and Fu, Peng and Ye, Min You and Zheng, Jin Xing and Lu, Kun and Gao, Xianggao and others},
  journal={IEEE Transactions on Plasma Science},
  volume={42},
  number={3},
  pages={503--509},
  year={2014},
  doi={10.1109/TPS.2014.2299277},
  publisher={IEEE}
}

@article{kube2022near,
  title={Near real-time streaming analysis of big fusion data},
  author={Kube, Ralph and Churchill, R Michael and Chang, CS and Choi, Jong and Wang, R and Klasky, Scott and Stephey, Laurie and Dart, E and Choi, MJ},
  journal={Plasma Physics and Controlled Fusion},
  volume={64},
  number={3},
  pages={035015},
  year={2022},
  doi={10.1088/1361-6587/ac3f42},
  publisher={IOP Publishing}
}

@article{Li2018An,
doi = {10.1088/1748-0221/13/10/P10029},
url = {https://dx.doi.org/10.1088/1748-0221/13/10/P10029},
year = {2018},
month = {oct},
publisher = {},
volume = {13},
number = {10},
pages = {P10029},
author = {Li, C. and Lan, T. and Wang, Y. and Liu, J. and Xie, J. and Lan, T. and Li, H. and Qin, H.},
title = {An automatic data cleaning procedure for electron cyclotron emission imaging on EAST tokamak using machine learning algorithm},
journal = {Journal of Instrumentation}
}

@article{Ferreira2018JETTomography,
title = {Deep learning for plasma tomography using the bolometer system at JET},
journal = {Fusion Engineering and Design},
volume = {114},
pages = {18-25},
year = {2017},
issn = {0920-3796},
doi = {https://doi.org/10.1016/j.fusengdes.2016.11.006},
url = {https://www.sciencedirect.com/science/article/pii/S0920379616306883},
author = {Francisco A. Matos and Diogo R. Ferreira and Pedro J. Carvalho}
}

@article{Guo2021EASTLSTM,
doi = {10.1088/1361-6587/ac228b},
url = {https://dx.doi.org/10.1088/1361-6587/ac228b},
year = {2021},
month = {sep},
publisher = {IOP Publishing},
volume = {63},
number = {11},
pages = {115007},
author = {Guo, B H and Chen, D L and Shen, B and Rea, C and Granetz, R S and Zeng, L and Hu, W H and Qian, J P and Sun, Y W and Xiao, B J},
title = {Disruption prediction on EAST tokamak using a deep learning algorithm},
journal = {Plasma Physics and Controlled Fusion}
}

@article{wei2023quantum,
title = {Quantum machine learning in medical image analysis: A survey},
journal = {Neurocomputing},
volume = {525},
pages = {42-53},
year = {2023},
issn = {0925-2312},
doi = {https://doi.org/10.1016/j.neucom.2023.01.049},
url = {https://www.sciencedirect.com/science/article/pii/S0925231223000589},
author = {Lin Wei and Haowen Liu and Jing Xu and Lei Shi and Zheng Shan and Bo Zhao and Yufei Gao}
}

@ARTICLE{ullah2024quantum,
  author={Ullah, Ubaid and Garcia-Zapirain, Begonya},
  journal={IEEE Access}, 
  title={Quantum Machine Learning Revolution in Healthcare: A Systematic Review of Emerging Perspectives and Applications}, 
  year={2024},
  volume={12},
  number={},
  pages={11423-11450},
  doi={10.1109/ACCESS.2024.3353461}}

@article{LAN2019159,
title = {Time-domain global similarity method for automatic data cleaning for multi-channel measurement systems in magnetic confinement fusion devices},
journal = {Computer Physics Communications},
volume = {234},
pages = {159-166},
year = {2019},
issn = {0010-4655},
doi = {https://doi.org/10.1016/j.cpc.2018.07.014},
url = {https://www.sciencedirect.com/science/article/pii/S0010465518302674},
author = {Ting Lan and Jian Liu and Hong Qin and Lin Li Xu}
}

@misc{brunton2021modern,
      title={Modern Koopman Theory for Dynamical Systems}, 
      author={Steven L. Brunton and Marko Budišić and Eurika Kaiser and J. Nathan Kutz},
      year={2021},
      eprint={2102.12086},
      archivePrefix={arXiv},
      primaryClass={math.DS},
      url={https://arxiv.org/abs/2102.12086}, 
}

@article{mezic2013analysis,
  title={Analysis of fluid flows via spectral properties of the Koopman operator},
  author={Mezi{\'c}, Igor},
  journal={Annual review of fluid mechanics},
  volume={45},
  number={1},
  pages={357--378},
  year={2013},
  doi={https://doi.org/10.1146/annurev-fluid-011212-140652},
  publisher={Annual Reviews}
}

@article{zhang2019dynamics,
  title={Dynamics reconstruction and classification via Koopman features},
  author={Zhang, Wei and Yu, Yao-Chi and Li, Jr-Shin},
  journal={Data mining and knowledge discovery},
  volume={33},
  number={6},
  pages={1710--1735},
  year={2019},
  doi={10.1007/s10618-019-00639-x},
  publisher={Springer}
}

@article{biamonte2017quantum,
  title={Quantum machine learning},
  author={Biamonte, Jacob and Wittek, Peter and Pancotti, Nicola and Rebentrost, Patrick and Wiebe, Nathan and Lloyd, Seth},
  journal={Nature},
  volume={549},
  number={7671},
  pages={195--202},
  year={2017},
  doi={https://doi.org/10.1038/nature23474},
  publisher={Nature Publishing Group UK London}
}

@article{benedetti2019parameterized,
  title={Parameterized quantum circuits as machine learning models},
  author={Benedetti, Marcello and Lloyd, Erika and Sack, Stefan and Fiorentini, Mattia},
  journal={Quantum science and technology},
  volume={4},
  number={4},
  pages={043001},
  year={2019},
  doi = {10.1088/2058-9565/ab4eb5},
  publisher={IOP Publishing}
}

@article{cerezo2022challenges,
  title={Challenges and opportunities in quantum machine learning},
  author={Cerezo, Marco and Verdon, Guillaume and Huang, Hsin-Yuan and Cincio, Lukasz and Coles, Patrick J},
  journal={Nature computational science},
  volume={2},
  number={9},
  pages={567--576},
  year={2022},
  doi={https://doi.org/10.1038/s43588-022-00311-3},
  publisher={Nature Publishing Group US New York}
}

@misc{li2024learning,
      title={Learning Parameterized Quantum Circuits with Quantum Gradient}, 
      author={Keren Li and Yuanfeng Wang and Pan Gao and Shenggen Zheng},
      year={2024},
      eprint={2409.20044},
      archivePrefix={arXiv},
      primaryClass={quant-ph},
      url={https://arxiv.org/abs/2409.20044}, 
}

@article{mcclean2018barren,
  title={Barren plateaus in quantum neural network training landscapes},
  author={McClean, Jarrod R and Boixo, Sergio and Smelyanskiy, Vadim N and Babbush, Ryan and Neven, Hartmut},
  journal={Nature communications},
  volume={9},
  number={1},
  pages={4812},
  year={2018},
  doi={https://doi.org/10.1038/s41467-018-07090-4},
  publisher={Nature Publishing Group UK London}
}

@article{liu2018quantum,
  title = {Quantum machine learning for quantum anomaly detection},
  author = {Liu, Nana and Rebentrost, Patrick},
  journal = {Phys. Rev. A},
  volume = {97},
  issue = {4},
  pages = {042315},
  numpages = {10},
  year = {2018},
  month = {Apr},
  publisher = {American Physical Society},
  doi = {10.1103/PhysRevA.97.042315},
  url = {https://link.aps.org/doi/10.1103/PhysRevA.97.042315}
}

@article{huang2021power,
  title={Power of data in quantum machine learning},
  author={Huang, Hsin-Yuan and Broughton, Michael and Mohseni, Masoud and Babbush, Ryan and Boixo, Sergio and Neven, Hartmut and McClean, Jarrod R},
  journal={Nature communications},
  volume={12},
  number={1},
  pages={2631},
  year={2021},
  doi={https://doi.org/10.1038/s41467-021-22539-9},
  publisher={Nature Publishing Group UK London}
}

@article{stilck2021limitations,
  title={Limitations of optimization algorithms on noisy quantum devices},
  author={Stilck Fran{\c{c}}a, Daniel and Garcia-Patron, Raul},
  journal={Nature Physics},
  volume={17},
  number={11},
  pages={1221--1227},
  year={2021},
  doi={https://doi.org/10.1038/s41567-021-01356-3},
  publisher={Nature Publishing Group UK London}
}

@book{mauroy2020koopman,
  title={The Koopman Operator in Systems and Control: Concepts, Methodologies, and Applications},
  author={Mauroy, A. and Mezi{\'c}, I. and Susuki, Y.},
  isbn={9783030357139},
  series={Lecture Notes in Control and Information Sciences},
  year={2020},
  address={Berlin/Heidelberg, Germany},
  publisher={Springer}
}

@article{schuld2021effect,
  title = {Effect of data encoding on the expressive power of variational quantum-machine-learning models},
  author = {Schuld, Maria and Sweke, Ryan and Meyer, Johannes Jakob},
  journal = {Phys. Rev. A},
  volume = {103},
  issue = {3},
  pages = {032430},
  numpages = {12},
  year = {2021},
  month = {Mar},
  publisher = {American Physical Society},
  doi = {10.1103/PhysRevA.103.032430},
  url = {https://link.aps.org/doi/10.1103/PhysRevA.103.032430}
}

@article{kates2019predicting,
  title={Predicting disruptive instabilities in controlled fusion plasmas through deep learning},
  author={Kates-Harbeck, Julian and Svyatkovskiy, Alexey and Tang, William},
  journal={Nature},
  volume={568},
  number={7753},
  pages={526--531},
  year={2019},
  doi          = {10.1038/s41586-019-1116-4},
  publisher={Nature Publishing Group UK London}
}

@article{karniadakis2021physicsinformed,
  author = {Karniadakis, George Em and Kevrekidis, Ioannis G. and Lu, Lu and Perdikaris, Paris and Wang, Sifan and Yang, Liu},
  doi = {10.1038/s42254-021-00314-5},
  issn = {25225820},
  journal = {Nature Reviews Physics},
  number = 6,
  pages = {422--440},
  refid = {Karniadakis2021},
  title = {Physics-informed machine learning},
  url = {https://doi.org/10.1038/s42254-021-00314-5},
  volume = 3,
  year = 2021
}

@article{Huang2022Quantum,
   title={Quantum advantage in learning from experiments},
   volume={376},
   ISSN={1095-9203},
   url={http://dx.doi.org/10.1126/science.abn7293},
   DOI={10.1126/science.abn7293},
   number={6598},
   journal={Science},
   publisher={American Association for the Advancement of Science (AAAS)},
   author={Huang, Hsin-Yuan and Broughton, Michael and Cotler, Jordan and Chen, Sitan and Li, Jerry and Mohseni, Masoud and Neven, Hartmut and Babbush, Ryan and Kueng, Richard and Preskill, John and McClean, Jarrod R.},
   year={2022},
   month=jun, pages={1182–1186} }

@article{Bharti2022Noisy,
  title = {Noisy intermediate-scale quantum algorithms},
  author = {Bharti, Kishor and Cervera-Lierta, Alba and Kyaw, Thi Ha and Haug, Tobias and Alperin-Lea, Sumner and Anand, Abhinav and Degroote, Matthias and Heimonen, Hermanni and Kottmann, Jakob S. and Menke, Tim and Mok, Wai-Keong and Sim, Sukin and Kwek, Leong-Chuan and Aspuru-Guzik, Al\'an},
  journal = {Rev. Mod. Phys.},
  volume = {94},
  issue = {1},
  pages = {015004},
  numpages = {69},
  year = {2022},
  month = {Feb},
  publisher = {American Physical Society},
  doi = {10.1103/RevModPhys.94.015004},
  url = {https://link.aps.org/doi/10.1103/RevModPhys.94.015004}
}

@article{Cerezo2021Variational,
   title={Variational quantum algorithms},
   volume={3},
   ISSN={2522-5820},
   url={http://dx.doi.org/10.1038/s42254-021-00348-9},
   DOI={10.1038/s42254-021-00348-9},
   number={9},
   journal={Nature Reviews Physics},
   publisher={Springer Science and Business Media LLC},
   author={Cerezo, M. and Arrasmith, Andrew and Babbush, Ryan and Benjamin, Simon C. and Endo, Suguru and Fujii, Keisuke and McClean, Jarrod R. and Mitarai, Kosuke and Yuan, Xiao and Cincio, Lukasz and Coles, Patrick J.},
   year={2021},
   month=aug, pages={625–644} }

@article{Lubasch2020Variational,
  title = {Variational quantum algorithms for nonlinear problems},
  author = {Lubasch, Michael and Joo, Jaewoo and Moinier, Pierre and Kiffner, Martin and Jaksch, Dieter},
  journal = {Phys. Rev. A},
  volume = {101},
  issue = {1},
  pages = {010301},
  numpages = {7},
  year = {2020},
  month = {Jan},
  publisher = {American Physical Society},
  doi = {10.1103/PhysRevA.101.010301},
  url = {https://link.aps.org/doi/10.1103/PhysRevA.101.010301}
}

@article{Herman2023Quantum,
   title={Quantum computing for finance},
   volume={5},
   ISSN={2522-5820},
   url={http://dx.doi.org/10.1038/s42254-023-00603-1},
   DOI={10.1038/s42254-023-00603-1},
   number={8},
   journal={Nature Reviews Physics},
   publisher={Springer Science and Business Media LLC},
   author={Herman, Dylan and Googin, Cody and Liu, Xiaoyuan and Sun, Yue and Galda, Alexey and Safro, Ilya and Pistoia, Marco and Alexeev, Yuri},
   year={2023},
   month=jul, pages={450–465} }

@article{Xie2025Neural,
doi = {10.1088/1361-6587/adba12},
url = {https://doi.org/10.1088/1361-6587/adba12},
year = {2025},
month = {mar},
publisher = {IOP Publishing},
volume = {67},
number = {4},
pages = {045001},
author = {Xie, Xiaoping and Lan, Ting and Liu, Haiqing and Zhu, Xiang and Mao, Wenzhe and Lan, Tao and Ding, Weixing},
title = {Neural-network based electron density profile inversion for interferometer on EAST tokamak},
journal = {Plasma Physics and Controlled Fusion}
}

@ARTICLE{Li2022Convolutional,
  author={Li, Zewen and Liu, Fan and Yang, Wenjie and Peng, Shouheng and Zhou, Jun},
  journal={IEEE Transactions on Neural Networks and Learning Systems}, 
  title={A Survey of Convolutional Neural Networks: Analysis, Applications, and Prospects}, 
  year={2022},
  volume={33},
  number={12},
  pages={6999-7019},
  doi={10.1109/TNNLS.2021.3084827}}

@article{Laurens2008Visualizing,
  author  = {Laurens van der Maaten and Geoffrey Hinton},
  title   = {Visualizing Data using t-SNE},
  journal = {Journal of Machine Learning Research},
  year    = {2008},
  volume  = {9},
  number  = {86},
  pages   = {2579--2605},
  url     = {http://jmlr.org/papers/v9/vandermaaten08a.html}
}

@article{Rousseeuw1987Silhouettes,
title = {Silhouettes: A graphical aid to the interpretation and validation of cluster analysis},
journal = {Journal of Computational and Applied Mathematics},
volume = {20},
pages = {53-65},
year = {1987},
issn = {0377-0427},
doi = {https://doi.org/10.1016/0377-0427(87)90125-7},
url = {https://www.sciencedirect.com/science/article/pii/0377042787901257},
author = {Peter J. Rousseeuw}
}

@INPROCEEDINGS{Furutanpey2023Architectural,
  author={Furutanpey, Alireza and Barzen, Johanna and Bechtold, Marvin and Dustdar, Schahram and Leymann, Frank and Raith, Philipp and Truger, Felix},
  booktitle={2023 IEEE International Conference on Quantum Software (QSW)}, 
  title={Architectural Vision for Quantum Computing in the Edge-Cloud Continuum}, 
  year={2023},
  volume={},
  number={},
  pages={88-103},
  doi={10.1109/QSW59989.2023.00021}}

@misc{farhi2018classification,
      title={Classification with Quantum Neural Networks on Near Term Processors}, 
      author={Edward Farhi and Hartmut Neven},
      year={2018},
      eprint={1802.06002},
      archivePrefix={arXiv},
      primaryClass={quant-ph},
      url={https://arxiv.org/abs/1802.06002}, 
}

@article{Schuld2015introduction,
author = {Maria Schuld and Ilya Sinayskiy and Francesco Petruccione},
title = {An introduction to quantum machine learning},
journal = {Contemporary Physics},
volume = {56},
number = {2},
pages = {172--185},
year = {2015},
publisher = {Taylor \& Francis},
doi = {10.1080/00107514.2014.964942},
URL = {https://doi.org/10.1080/00107514.2014.964942}
}

@misc{li2025realization,
      title={Realization of Thread Level Parallelism on Quantum Devices}, 
      author={Keren Li and Zidong Lin and Zheng An and Guanru Feng and Zipeng Wu and Shiyao Hou and Jingen Xiang},
      year={2025},
      eprint={2511.05436},
      archivePrefix={arXiv},
      primaryClass={quant-ph},
      url={https://arxiv.org/abs/2511.05436}, 
}

\newpage 
\begin{appendices}

\setcounter{figure}{0}
\renewcommand\thefigure{S\arabic{figure}}
\setcounter{table}{0}
\renewcommand\thetable{S\arabic{table}}

\section{Details for Data Processing and Feature Extraction Details}
\label{app:data}
The data processing pipeline is implemented in three sequential stages: raw signal standardization, Koopman-based linearization, and statistical feature extraction. The detailed protocols for each stage are described below.

\subsection{Raw Data Preprocessing and Labeling}

The input data consists of unstructured JSON files from the POINT diagnostic system. Each file contains time-series data for $11$ spatial channels (every discharge) and a time vector. First, we align the temporal axis by identifying the index corresponding to the plasma discharge onset ($t=0$). All pre-trigger samples ($t<0$) are discarded to focus on the active plasma phase. Ground truth are labeled manually. Channels listed in this field are labeled as $y=0$ (anomaly), while remaining channels are labeled as $y=1$ (normal). For the raw data baseline, we apply Z-score standardization to each valid channel sequence $\mathbf{x}(t)$. The standardized sequence $\mathbf{x}_{norm}(t)$ is calculated by subtracting the mean $\mu$ and dividing by the standard deviation $\sigma$, with a small epsilon term $\epsilon=10^{-6}$ added to the denominator to prevent division by zero in cases of flat signals. This ensures that the input to the deep learning models follows a standard normal distribution.

\subsection{Koopman Sequence and its Mathematical Formulation}
\label{app:koopman_math}
Koopman operator theory is employed to linearize the non-linear dynamics of fusion plasma diagnostics. The underlying hypothesis is that while the state space evolution of the plasma is highly non-linear, there exists an infinite-dimensional Hilbert space of observables where the evolution is linear. We approximate this operator using a data-driven approach akin to dynamic mode decomposition, focusing on the reconstruction error as a metric for anomaly detection.

\subsubsection*{Time-Delay Embedding and Hankel Matrix Construction}
Given a single-channel discrete time-series measurement $\mathbf{x} = [x_0, x_1, \dots, x_{T-1}]$ sampled at equivalent time steps, we first reconstruct the phase space geometry using time-delay embedding. This transforms the scalar observation into a high-dimensional state vector. We construct a Hankel matrix $H_m$ by stacking time-shifted copies of the data. For an embedding dimension $d_H$ and delay $\tau_H$, the matrix structure is defined as follows:
\begin{eqnarray}
H_m = \begin{bmatrix}
x_1 & x_{1+\tau_H} & \dots & x_{1+(d_H-1)\tau_H} \\
x_2 & x_{2+\tau_H} & \dots & x_{2+(d_H-1)\tau_H} \\
\vdots & \vdots & \ddots & \vdots \\
x_{k_H} & x_{k_H+\tau_H} & \dots & x_{k_H+(d_H-1)\tau_H}
\end{bmatrix}
\end{eqnarray}
where $k_H$ is the number of resulting delay vectors. This embedding ensures that the latent dynamical information hidden in the scalar signal is unfolded into a matrix form suitable for spectral analysis.

\subsubsection*{Singular Value Decomposition and Rank Truncation}
To extract the dominant coherent structures and filter out incoherent noise, we perform a Singular Value Decomposition on $H_m$. The decomposition is given by $H_m = U\Sigma V^*$, where $U$ and $V$ contain the left and right singular vectors, and $\Sigma$ is a diagonal matrix containing the singular values $\sigma_i$ arranged in descending order.

To obtain a low-rank approximation of the system dynamics, we truncate the decomposition to the first $r$ singular values. In our implementation, we set $r=11$, which empirically captures the primary energy modes of the plasma evolution while discarding high-frequency stochastic fluctuations. The reduced state coordinates $\mathbf{v}(t)$ are then obtained by projecting the system onto the dominant singular vectors:
\begin{equation}
\mathbf{v} = \Sigma_r V_r^*
\end{equation}
where $\Sigma_r$ and $V_r$ correspond to the truncated matrices. This step effectively maps the high-dimensional noisy input into a smooth, low-dimensional manifold governed by the dominant Koopman modes.

\subsubsection*{Linear Operator Approximation and Residual Calculation}
The core of the Koopman analysis lies in finding a linear operator $A$ that best describes the temporal evolution of the reduced state $\mathbf{v}(t)$. Unlike discrete-time maps, we adopt a continuous-time differential perspective. We first estimate the time derivative of the state vectors, $\dot{\mathbf{v}}(t)$, using a second-order central difference scheme.

We assume that within a local temporal window, the dynamics satisfy the linear differential equation $\dot{\mathbf{v}} \approx A \mathbf{v}$. The optimal linear operator $A$ is solved via a least-squares minimization problem:
\begin{equation}
A = \operatorname{argmin}_{A} \| \dot{\mathbf{v}} - A \bm{\mathbf{v}} \|_F
\end{equation}
The analytical solution is computed using the Moore-Penrose pseudoinverse, $A \approx \dot{\mathbf{v}} \mathbf{v}^\dagger$.

The Koopman reconstruction residual $R(t)$ is defined as the norm of the difference between the observed derivative and the linear prediction:
\begin{equation}
R(t) = \| \dot{\mathbf{v}}(t) - A \mathbf{v}(t) \|_2
\end{equation}
A low residual indicates that the plasma is in a coherent, stable regime well-described by linear Koopman dynamics. Conversely, a high residual signifies strong non-linearities. 

\subsection{Statistical Feature Extraction}
For the compact quantum classifier, we further compress the Koopman sequence. Since the length of the residual sequence varies with the discharge duration, we calculate the statistics that capture the distribution of the dynamical errors. The extracted features are for example, the mean, standard deviation, minimum value, maximum value, skewness, and kurtosis. These descriptors form the input vector $\mathbf{z}$ for PQNN. This statistical summarization bridges the dimensionality gap between the high-frequency diagnostic data and the limited qubit capacity of current quantum processors, while preserving the essential information regarding signal stability and anomaly distribution.

\section{Structure and Implementation of PQNN}
\label{app:PQNN}
The PQNN model integrates a classical projection layer with parallel PQCs to perform end-to-end classification on the features from Koopman sequence. In the forward pass, the input vector \(\mathbf{z}\) is first normalized via a layer normalization operation and then mapped by a linear transformation without bias:
\begin{equation}
\mathbf{h} = W_{enc} \, \mathrm{LayerNorm}(\mathbf{z})+\mathbf{b},
\end{equation}
where \(W_{enc}\) denotes the weight matrix which is set fixed. The resulting feature vector \(\mathbf{h}\) is then split into a set of equal-length subvectors \(\mathbf{h}^{(1)}, \mathbf{h}^{(2)} \dots \). Each sub-vector is processed by an independent PQC.

Specifically, the sub-vector is angle-encoded via single-qubit Pauli-Y rotations on qubits initialized in the zero state. This is supposed to be followed by a PQC with entangling CNOT operations and parameterized single-qubit rotations. 
The quantum outputs are subjected to a second layer normalization to stabilize the distribution before measurement. Pauli-Z measurements on each qubit yield a expectation vector \(\mathbf{e}^{(j)}\) for the \(j\)-th sub-circuit. 
These features are aggregated into raw logits by averaging alternate elements to form the final 2-dimensional output.

During training, the model optimizes the cross-entropy loss.
In our simulation, dimension of feature is set as \(6\) and \(2\) parallel \(3\)-qubit PQCs are used. 
In our work, we choose a full parameterized 3-qubit unitary for simulation, where 64 parameters are fine-tuned. Optimization is performed using the AdamW optimizer with a learning rate of 0.01 for 120 epochs and a batch size of \(64\).

\section{Details for Visualization Methodology}
\label{app:visualization}

To evaluate the decision boundary quality and verify the topological distinctness of the learned features, we implemented a comprehensive visualization pipeline using the t-Distributed Stochastic Neighbor Embedding (t-SNE) algorithm. The feature extraction and dimensionality reduction protocols are described below.

\subsection*{Data Preprocessing and Feature Extraction}
The visualization inputs were categorized into three distinct levels of abstraction, each requiring specific preprocessing to ensure computational feasibility and numerical stability:
\begin{enumerate}[label=\arabic*., nosep, leftmargin=*, listparindent=1em]
\item Raw Data Space: Directly visualizing the high-frequency diagnostic waveforms ($T \approx 20,000$) poses significant computational challenges for manifold learning algorithms. To mitigate this, we applied a temporal downsampling strategy, reducing the effective sequence length while preserving the macroscopic pulse structure. Variable-length sequences were zero-padded to align with the maximum duration in the test batch, resulting in a flattened input vector $\mathbf{x}_{raw} \in \mathbb{R}^{T}$ for the t-SNE engine.
\item Koopman Sequence Space: This visualization represents the physics-informed intermediate space. Instead of using raw noisy data, we utilized the Koopman operator to reconstruct the signal trajectories based on the dominant spectral modes. This process effectively acts as a spectral filter, removing high-frequency stochastic noise while preserving the coherent plasma dynamics. The resulting reconstructed time-series vectors $\mathbf{x}_{\text{seq}} \in \mathbb{R}^{T}$ were used as the input for this projection, explaining the improved local clustering (S-Score 0.458) compared to raw data.
\item Latent Representation Spaces: For the neural network models, we extracted the latent representations from the penultimate layer, defined as the vector space immediately preceding the final class probability assignment.
For the Koopman-PQNN, specifically, we visualized the post-quantum aggregated features. This corresponds to the 2-dimensional expectation vector obtained after the variational quantum circuit and the subsequent layer normalization, but prior to the final softmax operation. This specific extraction point reflects the actual non-linear manifold constructed by the quantum ansatz.
For the classical CNN baselines, we extracted the high-dimensional flattened output of the final convolutional block to visualize the dense feature clustering.
\end{enumerate}

\subsection*{t-SNE Configuration and Quantitative Metrics}
To generate the 2D projections, we utilized the t-SNE implementation from \texttt{scikit-learn} with a specific hyperparameter configuration optimized for cluster separation:

\begin{enumerate}[label=\arabic*., nosep, leftmargin=*, listparindent=1em]
    \item Initialization: We employed Principal Component Analysis  initialization rather than random initialization. This preserves the global topological structure of the data during the non-convex optimization process, ensuring that the relative positioning of clusters is meaningful.
    \item Perplexity: The perplexity parameter was set to $50.0$. This relatively high value (compared to the default 30) forces the algorithm to consider a larger number of nearest neighbors, which is crucial for revealing the global separation between the "Stable" and "anomaly" manifolds rather than focusing solely on local variations.
    \item Early Exaggeration: We increased the early exaggeration factor to $20.0$. This amplifies the attractive forces between natural clusters during the early phase of optimization, creating tighter and more distinct visual groupings.
    \item Optimization: The algorithm was run for a maximum of 1,000 iterations using the Euclidean distance metric. The learning rate was set to 'auto' to ensure adaptive convergence based on the data density.
\end{enumerate}
Finally, to quantify the visual clustering quality observed in the t-SNE projections, we calculated the S-Score  for each resulting 2D embedding. The S-Score measures the ratio of intra-cluster distance to inter-cluster distance, providing a metric $S \in [-1, 1]$ where high values confirm that the model has successfully disentangled the complex input correlations into linearly separable regions.

\end{appendices}
\end{document}